\documentclass[a4paper]{article}
\usepackage[T1]{fontenc}
\usepackage{amsfonts,amsmath,amsthm,amssymb}
\usepackage{calc}
\usepackage{gensymb}
\usepackage[nottoc, notlof, notlot]{tocbibind}
\usepackage{fullpage}
\usepackage{mathrsfs}  
\usepackage{bigints}
\usepackage{indentfirst}
\usepackage{enumerate}
\usepackage{authblk}

\DeclareMathOperator{\tr}{Tr}

\let\liminf\relax
\DeclareMathOperator*{\liminf}{liminf}
\newcommand{\normt}[1]{\left\Vert #1\right\Vert_{(t)}}
\newcommand{\norm}[1]{\left\Vert #1\right\Vert}

\begin{document}

\newtheorem{theorem}{Theorem} [section]
\newtheorem{prop}[theorem]{Proposition} 
\newtheorem{defi}[theorem]{Definition} 
\newtheorem{exe}[theorem]{Exemple} 
\newtheorem{lemma}[theorem]{Lemma} 
\newtheorem{rem}[theorem]{Remark} 
\newtheorem{cor}[theorem]{Corollary} 
\newtheorem{conj}[theorem]{Conjecture}
\renewcommand\P{\mathbb{P}}
\newcommand\E{\mathbb{E}}
\newcommand\N{\mathbb{N}}
\newcommand\1{\mathbf{1}}
\newcommand\C{\mathbb{C}}
\newcommand\CC{\mathcal{C}}
\newcommand\M{\mathbb{M}}
\newcommand\R{\mathbb{R}}
\newcommand\U{\mathbb{U}}
\newcommand\A{\mathcal{A}}
\newcommand\B{\mathcal{B}}
\newcommand\D{\mathcal{D}}
\newcommand\qc{\Phi_{n}}
\renewcommand\i{\mathbf{i}}
\renewcommand\S{\mathcal{S}_{N,t}}
\renewcommand\L{\mathcal{L}_{\Phi}}
\renewcommand\d{\partial_i}
\renewcommand\.{\ .}
\renewcommand\,{\ ,}

\def\etc{,\dots ,}

\title{Concentration estimates for random subspaces of a tensor product, and application to Quantum Information Theory}

\date{}

\author[1]{Beno\^it Collins}
\author[1,2]{F\'elix Parraud}
\affil[1]{\small Department of Mathematics, Graduate School of Science, Kyoto University, Kyoto 606-8502, Japan.}
\affil[2]{\small Universit\'e de Lyon, ENSL, UMPA, 46 all\'ee d'Italie, 69007 Lyon.}

\maketitle

\begin{abstract}
Given a random subspace $H_n$ chosen uniformly in a tensor product of Hilbert spaces $V_n\otimes W$, we 
consider the collection $K_n$ of all singular values of all norm one elements of $H_n$ with respect to the tensor structure.
A law of large numbers has been obtained for this random set in the context of $W$ fixed and the dimension 
of $H_n,V_n$ tending to infinity at the same speed by Belinschi, Collins and Nechita.

In this paper, we provide measure concentration estimates in this context.
The probabilistic study of $K_n$ was motivated by important questions in Quantum Information Theory, and allowed to provide
the smallest known dimension for the dimension of an ancilla space allowing Minimum Output Entropy (MOE) violation. 
With our estimates, we are  able, as an application, to provide actual bounds for the dimension of spaces where violation of
MOE occurs.
\end{abstract} 

\vspace*{0,5cm}



\section{Introduction}

One of the most important questions in Quantum Information Theory (QIT) was to 
figure out whether one can find two quantum channels $\Phi_1$ and $\Phi_2$ such that
\begin{equation}
\label{4ineqqMOE}
H_{min}(\Phi_1\otimes\Phi_2)<H_{min}(\Phi_1)+H_{min}(\Phi_2),
\end{equation}
where $H_{min}$ is the Minimum Output Entropy (MOE), defined in section \ref{4sec:QIT}. This problem was solved by \cite{has}, in which it was proved that indeed two such quantum channels exist. \cite{has} builds on important preliminary work by \cite{hayden-winter} and references therein. 
This question is important for the following reason: the inequality 
 $H_{min}(\Phi_1\otimes\Phi_2)\le H_{min}(\Phi_1)+H_{min}(\Phi_2)$ is always true, and if the answer to 
 the question involving Equation \eqref{4ineqqMOE} was negative, this would mean 
 $H_{min}(\Phi_1\otimes\Phi_2)= H_{min}(\Phi_1)+H_{min}(\Phi_2)$ holds for any 
 two quantum channels $\Phi_1$ and $\Phi_2$, and it would have implied that the classical capacity of a quantum channel is equal to its Holevo capacity, and consequently it would give a systematic way to compute the classical capacity. 
 For more explanation we refer to \cite{CY2019}.
 Unfortunately, the existence of two quantum channels $\Phi_1$ and $\Phi_2$ satisfying Equation \eqref{4ineqqMOE}
 defeated this hope. On the other hand, it opens the quest for a better understanding of this superadditivity phenomenon.

Indeed, all proofs available so far are not constructive in the sense that constructions rely on the probabilistic method.
After the initial construction of \cite{has}, the probabilistic tools involved in the proof have been found to have deep relation with random matrix theory
in many respects, including large deviation principle \cite{bho}, Free probability \cite{BCN2}, convex geometry \cite{asw} and Operator Algebra \cite{C18}. The last two probably give the most conceptual proofs, and in particular convex geometry gives explicit numbers. Free probability gives the smallest numbers for the output dimension \cite{BCN2} but was unable to give estimates for the input dimension so far. More generally, 
a large violation is obtained with free probability in \cite{BCN2} (the value $H_{min}(\Phi_1)+H_{min}(\Phi_2)-H_{min}(\Phi_1\otimes\Phi_2)$ can get arbitrarily close to $\log 2$, which is much bigger than the constants obtained in \cite{has}), but it
relates to a law of large numbers obtained in \cite{BCN1} whose  speed of convergence was not 
explicit, and in turn, did not give any estimate on the smallest dimension of the input space.
In order to obtain explicit parameters, measure concentration estimates, ideally large deviation estimates, are required. 
From a theoretical point of view, this is the goal of this paper.
Our main results (Theorem \ref{4apr} and \ref{4thm:numerics}) give precise estimates for the probability of additivity violation and  the  dimension  of  the  violating  channel  in  a  natural  random channel model. The proof is based on the far reaching approach of \cite{deux} -- see as well \cite{un}.
As a corollary, we obtain the following important application in Quantum Information Theory:
\begin{theorem}[For the precise statement, see Theorem \ref{4thm:numerics}]
	There exist a quantum channel from $\M_{184\times 10^{52}}(\C)$ to $\M_{184}(\C)$ such that combined with its conjugate channel, it yields violation of the MOE, i.e. such that they satisfy the inequality \eqref{4ineqqMOE}.
\end{theorem}

The paper is organized as follows. After this introduction, section 2 is devoted to 
introducing necessary notations and state the main theorem.
Section 3 contains the proof of the main theorem, and section 4 contains application to Quantum Information Theory.

Acknowledgements: B.C. was supported by JSPS KAKENHI 17K18734, 17H04823 and 20K20882.
F.P. was supported by a JASSO fellowship and Labex Milyon (ANR-10-LABX-0070) of Universit\'e de Lyon. This work was initiated while the second author was doing his MSc under the supervision of Alice Guionnet and he would like to thank
her for insightful comments and suggestions on this work. The authors would also like to thank Ion Nechita for an interesting remark on the minimum dimension with respect to $k$.

\section{Notations and main theorem}

We denote by $H$ a Hilbert space, which we assume to be finite dimensional. 
$B(H)$ is the set of bounded linear operators on $H$, and 
$D(H)\subset B(H)$ is the collection of trace 1, positive operators -- known as
\emph{density matrices}.
In the case of matrices, we denote it by  $\D_k\subset \M_k(\C)$.

Let $d,k,n\in\N$, let $U$ be distributed according to the Haar measure on the unitary group of $\M_{kn}(\C)$, let $P_n$ be the canonical injection from $\C^d$ to $\C^{kn}$, that is the matrix with $kn$ lines and $d$ columns with $1$ on the diagonal and $0$ elsewhere. With $\tr_n$ the unnormalized trace on $\M_n(\C)$, for $d\leq kn$, we define the following random linear map,
\begin{equation}
\label{4fond}
\qc : X\in\M_d(\C) \mapsto id_k \otimes \tr_n(U P_n X P_n^* U^*) \in\M_k(\C).
\end{equation}
This map is trace preserving, linear and completely positive (i.e. for any $l\in\N^*$, with $id_l:\M_l(\C)\to\M_l(\C)$ the identity map, $\phi_n\otimes id_l$ is a positive map, that is a map who maps positive elements to positive elements) and as such, it is known as a quantum channel. Let $t\in [0,1]$. Let $(d_n)_{n\in\N}$ be an integer sequence such that $d_n\sim  tkn$, and define 
\begin{equation}\label{4def-kkt}
K_{n,k,t} = \qc(\D_{d_n}).
\end{equation}

\noindent There is a much more geometric definition of $K_{n,k,t}$ thanks to the following proposition. Actually, while the quantum channel \eqref{4fond} is random, we do not use this fact in the proof, and we could very well prove the same result for a quantum channel defined as in equation \eqref{4fond} but with a deterministic unitary matrix instead of a Haar unitary matrix $U$.

\begin{prop}
	\label{4defeq}
	We have,
	\begin{equation}
	\label{4secondedef}
	K_{n,k,t} = \{ X\in \D_k \ |\  \forall A\in \D_k, \tr_k(XA) \leq \norm{P_n^* U^* A\otimes I_n U P_n} \} .
	\end{equation}
	
	\noindent Besides for any $A\in \D_k$, $ \{ X\in K_{n,k,t} \ |\ \tr_k(XA) = \norm{P_n^* U^* A\otimes I_n U P_n} \}$ is non-empty.
\end{prop}

\begin{proof}
	Let $Y\in \mathcal{D}_{d_n}$, $A\in \D_k$, then
	\begin{align*}
	\tr_k(\qc(Y)A) &= \tr_{kn}(U P_n Y P_n^* U^* \cdot A\otimes I_n) \\
	&= \tr_d( \sqrt{Y} P_n^* U^* \cdot A\otimes I_n \cdot U P_n \sqrt{Y}) \\
	&\leq \tr_d(Y) \norm{P_n^* U^* \cdot A\otimes I_n\cdot U P_n} \\
	&= \norm{P_n^* U^* \cdot A\otimes I_n \cdot U P_n} .
	\end{align*}
	
	\noindent Let us write $E$ for the right-hand side of the equation \eqref{4secondedef}, we just showed that $ K_{n,k,t} = \qc(\mathcal{D}_{d_n}) \subset E $. Besides if $P_x$ is the orthogonal projection on the vector $x$, we have that 
	$$ \norm{P_n^* U^* \cdot A\otimes I_n\cdot  U P_n} = \max_{x\in \C^d} \tr_{d}(P_n^* U^* \cdot A\otimes I_n\cdot U P_n P_x) = \max_{x\in \C^d} \tr_{k}(A\ \qc(P_x) ). $$
	
	\noindent Thus, for every $\varepsilon >0$ and $A\in \D_k$, we can find an element of $K_{n,k,t}$ in $\{ X\in \D_k \ |\ \tr_k(XA) \geq \norm{P_n^* U^* A\otimes I_n U P_n} - \varepsilon \}$. By compactness of $K_{n,k,t}$, we can even find an element of $K_{n,k,t}$ in $\{ X\in \D_k\  |\ \tr_k(XA) = \norm{P_n^* U^* A\otimes I_n U P_n} \}$.
	
	If we see $E$ as a convex set of $\M_k(\C)_{sa}$ the set of self-adjoint matrices of size $k$, let $X\in E$ be an exposed point of $E$, that is there exists $A\in\M_k(\C)_{sa}$ and $C$ such that the intersection of $E$ and $\{ Y\in \M_k(\C)_{sa} \ |\ \tr_k(AY)=C \}$ is reduced to $\{X\}$ and that $E$ is included in $\{ Y\in \M_k(\C)_{sa} \ |\ \tr_k(AY) \leq C \}$. We have the following equality for $\lambda$ large enough since if $Y\in \D_k$, $\tr_k(Y)=1$,
	
	$$ \{ Y\in \D_k \ |\ \tr_k(AY)=C \} = \left\{ Y\in \D_k \ |\ \tr_k\left(\frac{A+\lambda I_k}{\tr_k(A+\lambda I_k)}Y \right) = \frac{C+\lambda}{\tr_k(A+\lambda I_k)} \right\} . $$
	
	\noindent Thus, we can find $B\in \D_k$ and $c$ such that such that the intersection of $E$ and $\{ Y\in \D_k \ |\ \tr_k(BY) = c \}$ is reduced to $\{X\}$ and that $E$ is included in $\{ Y\in \D_k \ |\ \tr_k(BY) \leq c \}$. To summarize:
	\begin{itemize}
		\item The intersection of $K_{n,k,t}$ and $\{ Y\in \D_k \ |\ \tr_k(BY) = \norm{P_n^* U^* \cdot B\otimes I_n \cdot U P_n} \}$ 
		is non-empty.
		\item $K_{n,k,t} \subset E$, so the intersection of $E$ and 
		$\{ Y\in \D_k \ |\ \tr_k(BY) = \norm{P_n^* U^* \cdot B\otimes I_n\cdot U P_n} \}$ is non-empty.
		\item The intersection of $E$ and $\{ Y\in \D_k \ |\ \tr_k(BY) = c \}$ is exactly $\{X\}$.
		\item $E$ is included in both $\{ Y\in \D_k \ |\ \tr_k(BY) \leq c \}$ and 
		$$\{ Y\in \D_k \ |\ \tr_k(BY) \leq \norm{P_n^* U^* \cdot B\otimes I_n\cdot  U P_n} \} .$$
		
	\end{itemize} 
	Hence it implies that $c = \norm{P_n^* U^* B\otimes I_n U P_n}$ and that $X\in K_{n,k,t}$. Thus we showed that the exposed point of $E$ belongs to $K_{n,k,t}$. By a result of Straszewicz (\cite{dens},theorem $18.6$) the set of exposed points is dense in the set of extremal points, so the set of extremal points of $E$ is included in $K_{n,k,t}$. Since $K_{n,k,t}$ is convex, $E$ is included in $K_{n,k,t}$.
\end{proof}

Thanks to Theorem 1.4 of \cite{collins_male}, we know that for any $A\in\M_k(\C)$, $\norm{P_{n}^* U^* \cdot A\otimes I_{n}\cdot  U P_{n}}$ converges almost 
surely towards a limit $\normt{A}$, which we now describe in terms of free probability (for the interested reader we refer 
to \cite{voiculescu-dykema-nica}, but a non expert reader can take $\lim_n\norm{P_{n}^* U^* \cdot A\otimes I_{n} \cdot U P_{n}}$  
as the definition of $\normt{A}$ without loss of generality).
For $A\in\M_k(\C)$, we set:
\begin{equation}
\normt{A}:=\norm{p_t A p_t},
\end{equation}

\noindent where on the right side we took the operator norm of $p_t A p_t$, with $p_t$ a self-adjoint projection of trace $t$, free from $A$. Consequently, we define
\begin{equation}
\label{4limit}
K_{k,t} = \{ X\in \D_k \ |\  \forall A\in \D_k, \tr_k(XA) \leq \normt{A} \} .
\end{equation}
Given their definition, it seems natural to say that $K_{n,k,t}$ converges towards $K_{k,t}$. However it is not quite as straightforward. The convergence for the Hausdorff distance was proved in \cite{BCN1}, Theorem 5.2. 
More precisely the authors proved that given a random subspace of size $d_n$, $F_{n,k,t}$ the collection of singular values of unit vectors in 
this subspace converges for the Hausdorff distance towards a deterministic set $F_{k,t}$. It turns out that $K_{n,k,t}$ 
(respectively $K_{k,t}$) is the convex hull of the self-adjoint matrices whose eigenvalues are in $F_{n,k,t}$ (respectively $F_{k,t}$). 
However we do not use this theorem and our paper is independent from \cite{BCN1}. The main result of this paper is a measure concentration estimate and can be stated as follows. 

\begin{theorem}
	\label{4apr}
	If we assume that $t\in[0,1]$, $k, n \in\N$, 	$(d_n)_{n\in\N}$ is an integer sequence such that $d_n\sim  tkn$ and $d_n\leq tkn$, then for any $\varepsilon>0$ and $\ n \geq 3^4 \times 2^{31}\times \ln^2(kn) \times k^3 \varepsilon^{-4}$,
	$$ \P\left(K_{n,k,t} \not\subset K_{k,t} + \varepsilon\right) \leq e^{k^2(\ln(3k^2 \varepsilon^{-1})) -\frac{n}{k}\times \frac{\varepsilon^2}{576} }, $$
	where $K_{k,t} + \varepsilon = \{ Y\in\D_k\ |\ \exists X\in K_{k,t}, \norm{X-Y}_2\leq \varepsilon \}$ with $\norm{M}_2 := \sqrt{\tr_{k}(M^*M)}$.
\end{theorem}

While this does not prove the convergence for the Hausdorff distance of $K_{n,k,t}$ towards $K_{k,t}$ since we do not study the probability that $K_{k,t} \not\subset K_{n,k,t} + \varepsilon$. We could adapt our proof to get this result, but without getting estimates with explicit constants which would be detrimental to our aim of finding explicit parameters for violation of the additivity of the MOE.

\section{Proof of main theorem}

We will combine this geometrical description with the following lemma to get an estimate.

\begin{prop}
	\label{4haussd}
	If we define $K_{k,t} + \varepsilon = \{ Y\in\D_k\ |\ \exists X\in K_{k,t}, \norm{X-Y}_2\leq \varepsilon \}$ with $\norm{M}_2 := \sqrt{\tr_{k}(M^*M)}$, then the following implication is true,
	$$ \forall A\in \D_k,\ \norm{P_n^* U^* A\otimes I_n U P_n} \leq \normt{A} + \frac{\varepsilon}{k}\quad \implies\quad K_{n,k,t} \subset K_{k,t} + \varepsilon .$$
\end{prop}

Before proving it, we need a small lemma on the structure of $K_{k,t}$.

\begin{lemma}
	\label{4nonempty}
	Let $A\in\D_k$, then $\{ X\in K_{k,t} \ |\ \tr_k(XA) = \normt{A} \}$ is non-empty.
\end{lemma}

\begin{proof}
	Thanks to Proposition \ref{4defeq} we know that for any $n$, $\{ X\in K_{n,k,t} \ |\ \tr_k(XA) = \norm{P_n^* U^* A\otimes I_n U P_n} \}$ is non-empty. Hence there exists $X_n$ such that:
	\begin{itemize}
		\item $\tr_k(X_nA) = \norm{P_n^* U^* A\otimes I_n U P_n} $,
		\item $\forall B\in\D_k$, $\tr_k(X_nB) \leq \norm{P_n^* U^* B\otimes I_n U P_n} $.
	\end{itemize}
	By compactness of $\D_k$, we can assume that $X_n$ converges towards a limit $X$. But then as we said in the previous section, thanks to Theorem 1.4 from \cite{collins_male}, $\norm{P_n^* U^* B\otimes I_n U P_n}$ converges towards $\normt{B}$. Thus $X$ is such that:
	\begin{itemize}
		\item $\tr_k(XA) = \normt{A} $,
		\item $\forall B\in\D_k$, $\tr_k(XB) \leq \normt{B} $.
	\end{itemize}
	That is, $X$ belongs to $\{ X\in K_{k,t} \ |\ \tr_k(XA) = \normt{A} \}$.
\end{proof} 

We can now prove Proposition \ref{4haussd}.

\begin{proof}[Proof of Proposition \ref{4haussd}]
	We assume that $K_{n,k,t} \not\subset K_{k,t} + \varepsilon$, then thanks to the compactness of  $K_{n,k,t}$ and $K_{k,t}$, we can find $X\in K_{k,t}$ and $Y\in K_{n,k,t}$ such that $\norm{X-Y}_2 > \varepsilon$, and $K_{k,t} \cap B(Y, \norm{X-Y}_2)$ is empty. We set $V = \frac{Y - X}{\norm{Y-X}_2}$, $A =\frac{1}{k} \left(V+I_k\right)$, then $A\in \D_k$. We are going to show that $\norm{P_n^* U^* A\otimes I_n U P_n} > \normt{A} + \frac{\varepsilon}{k}$. To do so we define
	$$ P_C = \left\{ B\in K_{k,t} \ \middle| \ \tr_k(AB) = \frac{C+1}{k} \right\}
	= \left\{ B\in K_{k,t} \ \middle| \ \tr_k\left(VB\right) = C \right\} .$$
	
	\noindent Let us assume that for $C>\tr_k(VX)$, $P_C$ is not empty, then let $S\in P_C$. We can write $C = \tr_k(V(X+tV))$ for some $t>0$, thus $\tr_k(VS) = \tr_k(V(X+tV))$, that is $\tr_k((Y-X)(X-S))= - t \norm{Y-X}_2$. Hence the following estimate:
	\begin{align*}
	\norm{Y-(\alpha X + (1-\alpha) S)}_2^2 &= \tr_k\Big( \big( Y-X + (1-\alpha) (X-S)\big)^2 \Big) \\
	&=\norm{Y-X}_2^2 - 2t(1-\alpha)\norm{Y-X}_2 + \bigcirc((1-\alpha)^2) .
	\end{align*}
	
	Consequently since $K_{k,t}$ is convex, for any $\alpha$, $\alpha X + (1-\alpha) S\in K_{k,t}$, thus for $1-\alpha$ 
	small enough we could find an element of $K_{k,t}$ in $B(Y, \norm{X-Y}_2)$. Hence the contradiction. Thus for 
	$C > \tr_k(VX)$, $P_C$ is empty. By Lemma \ref{4nonempty}, we get that $\frac{\tr_k(VX)+1}{k} \geq \normt{A}$. Next we define
	$$ Q_C = \left\{ B\in K_{n,k,t} \ \middle| \ \tr_k(AB) = \frac{C+1}{k} \right\}
	= \left\{ B\in K_{n,k,t} \ \middle| \ \tr_k\left(VB\right) = C \right\} .$$
	Then clearly for $C=\tr_k(VY)$, $Q_C$ is non-empty since $Y\in Q_{\tr_k(VY)}$. Hence thanks to the geometric definition \eqref{4secondedef} of $K_{n,k,t}$, we have that $ \frac{\tr_k(VY)+1}{k} \leq \norm{P_n^* U^* A\otimes I_n U P_n}$. 
	Thus we have,
	$$ \norm{P_n^* U^* A\otimes I_n U P_n} \geq  \frac{\tr_k(V(Y-X))}{k} + \normt{A} = \frac{\norm{Y-X}_2}{k} + \normt{A} > \frac{\varepsilon}{k} + \normt{A} .$$
\end{proof}

Actually with a very similar proof, we could even show that almost surely there exist $A\in\D_k$ such that 
$$ d_H(K_{n,k,t},K_{k,t}) = k\times \left| \norm{P_n^* U^* \cdot A\otimes I_n \cdot U P_n} - \normt{A} \right|, $$
where $d_H$ is the Hausdorff distance associated to the norm $\norm{.}_2$ which comes from the scalar product $(U,V)\mapsto\tr_k(UV)$.
However this result will not be useful in this paper since the absolute value would be detrimental for the computation of our estimate.
The following lemma is a rather direct consequence of the previous proposition.

\begin{lemma}
	\label{4ptilemme}
	We set $\M_k(\C)_{sa}$ the set of self-adjoint matrices of size $k$. Let $u>0$, let $\mathcal{S}_u = \{ u M\  |\ M\in\M_k(\C)_{sa},\ \forall i\geq j,\ \Re(m_{i,j})\in \left\{\N+ \frac{1}{2}\right\}\cap [0,\lceil u^{-1} \rceil],\ \forall i>j,\ \Im(m_{i,j})\in \left\{\N+ \frac{1}{2}\right\}\cap [0,\lceil u^{-1} \rceil]  \}$, let $P_{\D_k}$ be the convex projection on $\D_k$. Then with $u = \frac{\sqrt{2}\varepsilon}{3k^2}$,
	
	$$ \P\left(K_{n,k,t} \not\subset K_{k,t} + \varepsilon\right) \leq \sum_{M\in \mathcal{S}_u} \P\left( \norm{P_n^* U^* (P_{\D_k}M\otimes I_n) U P_n} > \normt{P_{\D_k}M} + \frac{\varepsilon}{3k} \right) .$$
\end{lemma}

\begin{proof}
	We immediately get from proposition \ref{4haussd} that 
	$$ \P\left(K_{n,k,t} \not\subset K_{k,t} + \varepsilon\right) \leq \P\left( \exists A\in \D_k,\ \norm{P_n^* U^* \cdot A\otimes I_n \cdot U P_n} > \normt{A} + \frac{\varepsilon}{k} \right) .$$
	
	Now, let $A\in \D_k$, by construction of $\mathcal{S}_u$, there exists $M\in \mathcal{S}_u$ such that the real part and the imaginary part of the coefficients of $M$ are $u/2$-close from those of $A$. Thus we have $\norm{A-M}_2 \leq \frac{ku}{\sqrt{2}}$. Hence if we fix $u = \frac{\sqrt{2}\varepsilon}{3k^2}$, then we can always find $M\in \mathcal{S}_u$ such that $\norm{A-M}_2 \leq \frac{\varepsilon}{3k}$. Besides we have,
	$$ \left| \normt{A} - \normt{P_{\D_k}M} \right| \leq \norm{A-P_{\D_k}M} \leq \norm{A-P_{\D_k}M}_2 ,$$
	$$ \Big| \norm{P_n^* U^* A\otimes I_n U P_n} - \norm{P_n^* U^* (P_{\D_k}M \otimes I_n) U P_n} \Big| \leq \norm{A-P_{\D_k}M} \leq \norm{A-P_{\D_k}M}_2 .$$
	Hence since $P_{\D_k}A = A$ and that $P_{\D_k}$ is $1$-lipschitz as a function on $\M_k(\C)$ endowed with the norm $\norm{.}_2$, we have $\norm{A-P_{\D_k}M}_2 \leq \norm{A-M}_2 \leq \frac{\varepsilon}{3k}$. Consequently,
	$$ \left\{ \norm{P_n^* U^* \cdot A\otimes I_n\cdot  U P_n} > \normt{A} + \frac{\varepsilon}{k} \right\} \subset \left\{ \norm{P_n^* U^* (P_{\D_k}M\otimes I_n) U P_n} > \normt{P_{\D_k}M} + \frac{\varepsilon}{3k} \right\} .$$
	Hence,
	\begin{align*}
	&\left\{ \exists A\in \D_k,\ \norm{P_n^* U^* \cdot A\otimes I_n\cdot U P_n} > \normt{A} + \frac{\varepsilon}{k} \right\} \\
	&\subset \bigcup_{M\in \mathcal{S}_u} \left\{ \norm{P_n^* U^* (P_{\D_k}M\otimes I_n) U P_n} > \normt{P_{\D_k}M} + \frac{\varepsilon}{3k} \right\} .
	\end{align*}
	The conclusion follows.
	
\end{proof}

The next lemma  shows that there exist a smooth function which verifies some assumptions on the infinite norm of its derivatives.

\begin{lemma}
	\label{4fonctionsmooth}
	There exists $g$ a $\mathcal{C}^{6}$ function which takes value $0$ on $(-\infty,0]$ and value $1$ on $[1,\infty)$, and in $[0,1]$ otherwise. Besides for any $j\leq 6$, $\norm{g^{(j)}}_{\infty} = 2^{\frac{j(j+1)}{2}}$.
\end{lemma}

\begin{proof}
	Firstly we define,
	$$ f:t\in[0,1]\mapsto \left\{
	\begin{array}{ll}
	2t & \text{ if } t\leq 1/2 \\
	2(1-t) & \text{ if } t\geq 1/2
	\end{array} \right., $$
	$$
	\begin{array}{ccccc}
	H & : & \mathbb{C}^0([0,1]) & \to & \mathbb{C}^0([0,1]) \\
	& & f & \mapsto & t \mapsto \left\{
	\begin{array}{ll}
	f(2t) & \text{ if } t\leq 1/2 \\
	-f(2t-1) & \text{ if } t\geq 1/2
	\end{array}
	\right. \\
	\end{array} .
	$$
	Inspired by Taylor's Theorem, we define 
	$$ h : x\in[0,1] \mapsto \int_{0}^{x} \frac{(x-t)^5}{5!}\ H^5f(t)\ dt . $$
	It is easy to see that $h\in\mathcal{C}^6([0,1])$ with 
	$$ \forall j\leq 5,\quad h^{(j)} : x\in[0,1] \mapsto \int_{0}^{x} \frac{(x-t)^{5-j}}{(5-j)!} H^5f(t)\ dt,\quad\quad h^{(6)} = H^5f  .$$
	Thus one can easily extend $h$ by $0$ on $\R^{-}$ and $h$ remains $\mathcal{C}^6$ in $0$, as for what happens in $1$ it is way less obvious. In order to build $g$ we want to show that
	$$ \forall 1\leq j\leq 6,\quad h^{(j)}(1) =0,\quad\quad h(1)>0 .$$
	To do so let $w\in\mathcal{C}^0([0,1])$, then for any $k\geq 0$,
	\begin{align*}
	\int_0^1 (1-t)^k Hw(t) dt &= \int_0^{1/2} (1-t)^k w(2t) dt - \int_{1/2}^{1} (1-t)^k w(2t-1) dt \\
	&= \frac{1}{2^{k+1}} \int_0^1 \left( \left(2-t\right)^k - \left(1-t\right)^k\right) w(t) dt \\
	&= \frac{1}{2^{k+1}} \int_0^1 \sum_{0\leq i <k} \binom{k}{i} (1-t)^i w(t) dt .\\		 
	\end{align*}
	Thus recursively one can show that $\forall 1\leq j\leq 6$, $h^{(j)}(1) =0$. We also get that 
	$$ h(1) = \int_{0}^{x} \frac{(1-t)^5}{5!}\ H^5f(t)\ dt = 2^{-\sum_{2\leq i \leq 6 } i}\int_{0}^{x} f(t)\ dt = 2^{-21} .$$
	Hence we fix $g = 2^{21}h$, further studies show that $\norm{g^{(j)}}_{\infty} = 2^{\frac{j(j+1)}{2}}$.
\end{proof}

In the next lemma, we prove a first rough estimate on the deviation of the norm with respect to its limit. It is the only one where we use that $d_n\leq tkn$. 

\begin{lemma}
	\label{4firstestimate}
	For any $A\in\D_k$, $\varepsilon>0$,
	\begin{equation}
	\label{4beline}
	\P\left( \norm{P_n^* U^* A\otimes I_n U P_n} \geq \normt{A} + \varepsilon \right) \leq 3\times 2^{23 }\times \frac{\ln^2(kn)}{kn} \varepsilon^{-4} .
	\end{equation}
	
\end{lemma}

\begin{proof}
	
	For a better understanding of the notations and tools used in this proof, such as free stochastic calculus, we refer to \cite{deux}. 
	In particular $\tau_{kn}$ is the trace on the free product of $\M_{kn}(\C)$ with a $\CC^*$-algebra which contains a free unitary 
	Brownian motion, see Definition 2.8 of \cite{deux}. As for $\delta, \D$ and $\boxtimes$, see Definition 2.5 of \cite{deux} and 2.10 from \cite{trois}. If you are not familiar with free probability, it is possible to simply admit equation \eqref{4int1} to avoid having to understand the previous notations.
	
	Since $\norm{P_n^* U^* \cdot A\otimes I_n\cdot  U P_n} = \norm{P_nP_n^* U^*\cdot  A\otimes I_n\cdot U P_nP_n^*}$, we will 
	rather work with $P = P_nP_n^*$ since it is a square matrix. To simplify notations, instead of $A\otimes I_n$ we simply write $A$. 
	Let us now consider a function $f$ such that
	\begin{equation}
	\label{4fourrierhyp}
	\forall x\in\R,\quad f(x) = \int_{\R} e^{\i xy} d\mu(y),
	\end{equation}
	for some measure $\mu$. We set $v_s=U\otimes I_l \times V^{lkn}_{r-s}U^{lkn}_s$ and $w_s=U\otimes I_l \times W^{lkn}_{r-s}U^{lkn}_s$ where $(U^{mkn}_r)_{r\geq 0}$, $(V^{mkn}_r)_{r\geq 0}$ and $(W^{mkn}_r)_{r\geq 0}$ are independent unitary Brownian motions of size $mkn$ started in the identity. We also set $P_{l,l'} = I_{kn}\otimes E_{l,l'}$ where $E_{l,l'}$ is the matrix of size $m$ whom all coefficients are zero but the $(l,l')$ one which is $1$. Then thanks to Lemma 4.2, 4.6 and Corollary 3.3 from \cite{deux}, with $u_T$ a free unitary Brownian motion at time $T$ started in $1$,we have the following expression,
	\begin{align*}
	&\E\left[ \frac{1}{kn} \tr_{kn}\Big( f(P\ U^*A\ U\ P) \Big) \right] - \E\left[\tau_{kn}\Big( f(P\ u^*_T U^*\ A\ U u_T\ P)\Big)\right] \\
	&= \lim_{m\to\infty} \frac{1}{2m^2(kn)^3} \sum_{ 1\leq l,l'\leq m } \int \int_0^{T} \int_0^r \tr_{mkn}\Bigg( \delta\circ \delta^1 \circ \D \Big( e^{\i y Pv_s^*Av_sP} \Big)\widetilde{\#} P_{l',l}\  \\
	&\quad\quad\quad\quad\quad\quad\quad\quad\quad\quad\quad\quad\quad\quad\quad\quad\quad\quad \boxtimes\ \delta\circ \delta^2 \circ \D \Big( e^{\i y Pw_s^*Aw_sP} \Big)\widetilde{\#} P_{l,l'} \Bigg) ds dr\ d\mu(y) .
	\end{align*}
	Then if we set $R_1^s = Pv_s^*Av_sP$ and $R_2^s = Pw_s^*Aw_sP$ , after a lengthy computation,
	
	\begin{align*}
	&\tr_{mkn}\Bigg( \delta\circ \delta^1 \circ \D \Big( e^{\i y Pv_s^*Av_sP} \Big)\widetilde{\#} P_{l',l}\ \boxtimes\ \delta\circ \delta^2 \circ \D \Big( e^{\i y Pw_s^*Aw_sP} \Big)\widetilde{\#} P_{l,l'} \Bigg) \\
	=& - \i y \tr_{lkn}\Big( \delta (v_s^* A v_s)\widetilde{\#}P_{l',l} \times \delta(P e^{\i y R_2^s} P) \widetilde{\#}P_{l,l'} \Big) \nonumber \\
	& - \i y  \tr_{lkn}\Big( \delta(P e^{\i y R_1^s} P) \widetilde{\#}P_{l',l} \times \delta (w_s^* A w_s)\widetilde{\#}P_{l,l'} \Big) \nonumber \\
	& + y^2 \int_0^1 \tr_{lkn}\Big(  \delta(v_s^* A v_s P e^{\i y\alpha R_1^s} P v_s^* A v_s) \widetilde{\#}P_{l',l} \times \delta (P e^{\i y(1-\alpha) R_2^s} P)\widetilde{\#}P_{l,l'} \Big)\ d\alpha \nonumber \\
	& - y^2 \int_0^1 \tr_{lkn}\Big(  \delta(v_s^* A v_s P e^{\i y\alpha R_1^s} P) \widetilde{\#}P_{l',l} \times \delta ( w_s^* A w_sP e^{\i y(1-\alpha) R_2^s} P)\widetilde{\#}P_{l,l'} \Big)\ d\alpha \nonumber \\
	& + y^2 \int_0^1 \tr_{lkn}\Big(  \delta(P e^{\i y\alpha R_1^s} P) \widetilde{\#}P_{l',l} \times \delta ( w_s^* A w_sP e^{\i y(1-\alpha) R_2^s} P w_s^* A w_s )\widetilde{\#}P_{l,l'} \Big)\ d\alpha \nonumber \\
	& - y^2 \int_0^1 \tr_{lkn}\Big(  \delta(P e^{\i y\alpha R_1^s} P v_s^* A v_s) \widetilde{\#}P_{l',l} \times \delta ( P e^{\i y(1-\alpha) R_2^s} P w_s^* A w_s)\widetilde{\#}P_{l,l'} \Big)\ d\alpha \nonumber
	%&=  y^2 \tr_{lkn}\Bigg( (P_{l',l} v_s^*Av_s - v_s^*Av_s P_{l',l}) \nonumber\\
	%&\quad\quad\quad\quad\quad \times \int_0^1 \big( P e^{\i y R_2^s(1-\alpha)}P P_{l,l'} P e^{\i y R_2^s \alpha }P w_s^* A w_s - w_s^* A w_s P e^{\i y R_2^s(1-\alpha)}P P_{l,l'} P e^{\i y R_2^s \alpha }P \big)\ d\alpha \Bigg)   \nonumber\\
	%&+  y^2 \tr_{lkn}\Bigg( \int_0^1 \big( P e^{\i y R_1^s(1-\alpha)}P P_{l',l} P e^{\i y R_1^s \alpha }P v_s^* A v_s - v_s^* A v_s P e^{\i y R_1^s(1-\alpha)}P P_{l',l} P e^{\i y R_1^s \alpha }P \big)\ d\alpha   \nonumber\\
	%&\quad\quad\quad\quad\quad \times (P_{l,l'} w_s^*Aw_s - w_s^*Aw_s P_{l,l'}) \Bigg) \nonumber\\
	\end{align*}
	
	Since the norm of $A,P,v_s$ and $w_s$ are smaller than $1$, and that the rank of $P_{l,l'}$ is $kn$, by using the fact that the non-renormalized trace of a matrix of norm smaller than $1$ is smaller or equal to its rank, we finally get that for any $r$ and $s$,
	\begin{align*}
	&\frac{1}{kn}\left| \tr_{mkn}\Bigg( \delta\circ \delta^1 \circ \D \Big( e^{\i y Pv_s^*Av_sP} \Big)\widetilde{\#} P_{l',l}\ \boxtimes\ \delta\circ \delta^2 \circ \D \Big( e^{\i y Pw_s^*Aw_sP} \Big)\widetilde{\#} P_{l,l'} \Bigg) \right|\\
	&\leq 8y^2 + 2y^2 \int_0^1 (4+2\alpha |y|)\times 2(1-\alpha)|y|\ d\alpha + 2y^2 \int_0^1 (2+2\alpha |y|)\times (2+2(1-\alpha)|y|)\ d\alpha \\
	&= 16 y^2 + 16 |y|^3 + \frac{8}{3}y^4 .
	%&= 4 y^2 + 6 |y|^3 + \frac{4}{3}y^4 .
	\end{align*}
	
	\noindent Consequently, we get that 
	\begin{align*}
	&\left|\E\left[ \frac{1}{kn} \tr_{kn}\Big( f(P\ U^*A\ U\ P) \Big) \right] - \E\left[\tau_{kn}\Big( f(P\ u^*_T U^*\ A\ U u_T\ P)\Big)\right]\right| \\
	&\leq \frac{T^2}{(kn)^2} \int 4 y^2 + 4 |y|^3 + \frac{2}{3}y^4\ d|\mu|(y).
	\end{align*}
	
	\noindent Thanks to Proposition 3.2 from \cite{deux}, we get that for any $T\geq 5$, there exist a free Haar unitary $\widetilde{u}_T$ such that
	\begin{align*}
	&\left| \tau_{kn}\Big( e^{\i y P u^*_TU^* A U u_T P}\Big) - \tau\Big( e^{\i y P \widetilde{u}_T^* A \widetilde{u}_T P}\Big) \right| \\
	&=\left| \tau_{kn}\Big( e^{\i y P u^*_TU^* A U u_T P}\Big) - \tau\Big(  e^{\i y P \widetilde{u}_T^*U^* A U \widetilde{u}_T P}\Big) \right| \\
	&= \left| y \int_{0}^{1} \tau\left( e^{\i \alpha y P u^*_T A u_T P}P\ (u^*_T U^*AU u_T - \widetilde{u}^*_T U^*AU \widetilde{u}_T) P\  e^{\i (1-\alpha) y P\ \widetilde{u}^*_T A \widetilde{u}_T P} \right) d\alpha\right| \\
	&\leq 8 e^2 \pi e^{-T/2} |y| .
	\end{align*}
	Consequently with $u$ a free Haar unitary, if the support of $f$ and the spectrum of $P u^*A u P$ are disjoint, then $\tau\Big( f(P\ u^*\ A\otimes I_n\ u\ P)\Big) =0$, and 
	\begin{align}
	\label{4int1}
	&\left| \E\left[ \frac{1}{kn} \tr_{kn}\Big( f(P\ U^*\ A\otimes I_n\ U\ P) \Big) \right] \right| \\
	&\leq 8 e^2 \pi e^{-T/2} \int |y|\ d|\mu|(y) + \left(\frac{T}{kn}\right)^2 \int 4 y^2 + 4 |y|^3 + \frac{2}{3}y^4\ d|\mu|(y) . \nonumber
	\end{align}
	Let $g$ be a $\mathcal{C}^{6}$ function which takes value $0$ on $(-\infty,0]$ and value $1$ on $[1,\infty)$, and in $[0,1]$ otherwise. 
	We set $f_{\varepsilon}:x\mapsto g(2\varepsilon^{-1} (x - \alpha) - 1)g(2\varepsilon^{-1} (1-x)+1)$ with $\alpha = \normt{A}$. Then the 
	support of $f_{\varepsilon}$ is included in $[\normt{A},\infty)$, whereas the spectrum of $P u^*A u P$ is bounded by 
	$\norm{P u^*A u P} = \norm{A}_{(d_n(kn)^{-1})} \leq \normt{A}$ since $d_n\leq tkn$. Hence $f_{\varepsilon}$ satisfies \eqref{4int1}. 
	Setting $h:x\mapsto g(x - 2 \varepsilon^{-1}\alpha -1) g( 2\varepsilon^{-1} +1  -x)$, we have with convention 
	$\widehat{f}(x) = (2\pi)^{-1} \int_{\R} f(y) e^{-\i xy} dy$, for $1\leq k\leq 4$ and any $\beta>0$,	
	\begin{align*}
	\int |y|^k |\hat{f_{\varepsilon}}(y)|\ dy &= \frac{1}{2\pi} \int |y|^k \left|\int g(2\varepsilon^{-1} (x - \alpha) - 1)g(2\varepsilon^{-1} (1-x)+1) e^{-\i y x}\ dx \right|\ dy \\
	&= \frac{1}{2\pi} \int |y|^k \left|\int h(\beta x) e^{-\i y \varepsilon\beta x /2}\ \frac{\varepsilon\beta}{2} dx \right|\ dy \\
	&= \frac{1}{2\pi} 2^k\varepsilon^{-k} \beta^{-k} \int |y|^k \left|\int h(\beta x) e^{-\i y x}\ dx \right|\ dy \\
	&\leq \frac{1}{2\pi} 2^k\varepsilon^{-k} \int \frac{1}{1+y^2}\ dy \int ( |h^{(k)}(\beta x)| + \beta^2 |h^{(k+2)}(\beta x)|)\ dx \\
	&\leq 2^{k}\varepsilon^{-k} \left(\beta^{-1}\norm{g^{(k)}}_{\infty} + \beta \norm{g^{(k+2)}}_{\infty}\right) \.
	\end{align*}
	In the last line we used the fact that we can always assume that $\alpha+\varepsilon\leq 1$ (otherwise 
	$$\P( \norm{P_n^* U^* \cdot A\otimes I_n\cdot U P_n} \geq \normt{A} + \varepsilon ) = 0$$ and there is no need to do any 
	computation). Consequently $2\alpha \varepsilon^{-1} +2 \leq 2\varepsilon^{-1}$, and since the support of the derivatives of $x\mapsto g(x - 2 \varepsilon^{-1}\alpha -1)$ and those of $x\mapsto g( 2\varepsilon^{-1} +1  -x)$ are respectively included in $[2\alpha \varepsilon^{-1} +1,2\alpha \varepsilon^{-1} +2]$ and $[2 \varepsilon^{-1},2\varepsilon^{-1} +1]$, they are disjoint. Thus by fixing $\beta = \sqrt{\norm{g^{(k)}}_{\infty} \norm{g^{(k+2)}}_{\infty}^{-1}}$ 
	we get 
	\begin{align*}
	\int |y|^k |\hat{f_{\varepsilon}}(y)|\ dy &\leq 2^{k+1}\varepsilon^{-k} \sqrt{\norm{g^{(k)}}_{\infty} \norm{g^{(k+2)}}_{\infty}} \.
	\end{align*} 
	Consequently, since $f_{\varepsilon}$ satisfies \eqref{4fourrierhyp} with $d\mu(y) = \widehat{f_{\varepsilon}}(y) dy$, by using \eqref{4int1} we get 
	
	\begin{align*}
	&\quad \left| \E\left[ \frac{1}{kn} \tr_{kn}\Big( f_{\varepsilon}(P\ U^*\ A\otimes I_n\ U\ P) \Big) \right] \right| \nonumber \\
	&\leq 2^5 e^2 \pi e^{-T/2} \sqrt{\norm{g^{(1)}}_{\infty} \norm{g^{(3)}}_{\infty}} \varepsilon^{-1} + \left(\frac{T}{kn}\right)^2 2^5\varepsilon^{-2} \sqrt{\norm{g^{(2)}}_{\infty} \norm{g^{(4)}}_{\infty}} \\
	&\quad + \left(\frac{T}{kn}\right)^2 2^6\varepsilon^{-3} \sqrt{\norm{g^{(3)}}_{\infty} \norm{g^{(5)}}_{\infty}} + \left(\frac{T}{kn}\right)^2\frac{2^6}{3}\varepsilon^{-4} \sqrt{\norm{g^{(4)}}_{\infty} \norm{g^{(6)}}_{\infty}}. \nonumber
	\end{align*}
	Combined with Lemma \ref{4fonctionsmooth} and fixing $T = 4\ln(kn)$, we get
	
	\begin{align*}
	&\quad \left| \E\left[ \frac{1}{kn} \tr_{kn}\Big( f_{\varepsilon}(P\ U^*\ A\otimes I_n\ U\ P) \Big) \right] \right| \\
	&\leq 2^{17/2} e^2 \pi\ \frac{\varepsilon^{-1}}{(kn)^2} + 2^{31/2} \left(\frac{\ln(kn)}{kn}\right)^2 \varepsilon^{-2} + 2^{41/2}\left(\frac{\ln(kn)}{kn}\right)^2 \varepsilon^{-3} + \frac{2^{51/2}}{3} \left(\frac{\ln(kn)}{kn}\right)^2 \varepsilon^{-4} .
	\end{align*}
	Since for any $n$, almost surely $\norm{P_n^* U^* \cdot A\otimes I_n\cdot U P_n} \leq 1$, we have 
	\begin{align*}
	&\P\left( \norm{P_n^* U^* A\otimes I_n U P_n} \geq \normt{A} + \varepsilon \right) \\
	&= \P\Big( \exists \lambda\in \sigma(P U^* A\otimes I_n U P),\ f_{\varepsilon}(\lambda) = 1 \Big) \\
	&\leq \P\Big( \tr_{kn}\Big(f_{\varepsilon}(P^* U^* A\otimes I_n U P)\Big) \geq 1 \Big) \\
	&\leq \E\left[\tr_{kn}\Big( f_{\varepsilon}(P\ U^*\ A\otimes I_n\ U\ P) \Big) \right] \\
	&\leq 2^{\frac{17}{2}} e^2 \pi\ \frac{\varepsilon^{-1}}{kn} + \frac{\ln^2(kn)}{kn} \left(2^{31/2} \varepsilon^{-2} + 2^{41/2} \varepsilon^{-3} + \frac{2^{51/2}}{3} \varepsilon^{-4} \right) .
	\end{align*}
	One can always assume that $\ln^2(kn)\geq 1$ since for small value of $k$ and $n$, \eqref{4beline} is easily verified since the right-hand side of the inequality is larger than $1$. One can also assume that $\varepsilon<1$ since almost surely $\norm{P_n^* U^* A\otimes I_n U P_n}\leq 1$. We get the conclusion by a numerical computation.
	
\end{proof}

We can now refine this inequality by relying on corollary 4.4.28 of \cite{alice}, we state the part that we will be using in the next proposition.

\begin{prop}
	\label{4hyp}
	We set $S\U_N = \{X\in \U_N \ |\ \det(X)=1 \}$, let $f$ be a continuous, real-valued function on $\U_N$. We assume that there exists a constant $C$ such that for every $X,Y\in\U_N$,
	
	\begin{equation}
	|f(X)-f(Y)|\leq C \norm{X-Y}_2
	\end{equation}
	Then if we set $\nu_{S\U_N}$ the law of the Haar measure on $S\U_N$, with $U$ a Haar unitary matrix, for all $\delta>0$:
	\begin{equation}
	\P\left( \left|f(U) - \int f(YU) d\nu_{S\U_N}(Y) \right| \geq \delta\right) \leq 2 e^{-\frac{N\delta^2}{4C^2}}
	\end{equation}
\end{prop}

\begin{lemma}
	\label{4sndestimate}
	For any $A\in\D_k$, $\varepsilon>0$, if $\ kn \geq 2^{31}\times \ln^2(kn) \times \varepsilon^{-4}$, we have 
	\begin{equation*}
	\P\left( \norm{P_n^* U^* \cdot A\otimes I_n\cdot  U P_n} \geq \normt{A} + \varepsilon \right) \leq 2 e^{-kn\times \frac{\varepsilon^2}{64}} .
	\end{equation*}
	
\end{lemma}

\begin{proof}
	We set,
	$$ f : X\in \U_n \mapsto \norm{P_n^*X^* A\otimes I_n XP_n}, $$
	$$h : X\in \U_n \mapsto \int f(YX) d\nu_{S\U_{kn}}(Y) .$$
	If $U^1$ is a random matrix of law $\nu_{S\U_{kn}}$, and $\alpha$ a scalar of law $\nu_{\U_1}$ independent of $U^1$. Then the law of $\alpha U^1$ is $\nu_{\U_{kn}}$ since its law is invariant by multiplication by a unitary matrix. Consequently for any $X\in \U_{kn}$,
	
	$$ h(X) = \E[f(U^1X)] = \E[f(\alpha U^1 X)] = \E[f(\alpha U^1)] = \E[f(U)] .$$
	The second equality is true since for any scalar $\alpha$ and $X\in \U_{kn}$, $f(X) = f(\alpha X)$, and in the third one we used the invariance of the Haar measure on the unitary group by multiplication by a unitary matrix. Besides we also have that for any $X,Y\in \U_{kn}$,
	$$ |f(X) - f(Y)| \leq 2 \norm{X-Y} \leq 2 \norm{X-Y}_2 .$$
	Thus by using Proposition \ref{4hyp}, we get 
	$$ \P\left( \left| \norm{P_n^*U^* \cdot A\otimes I_n\cdot UP_n} - \E\Big[ \norm{P_n^*U^* A\otimes I_n UP_n} \Big] \right| \geq \delta \right) \leq 2 e^{-\frac{kn\delta^2}{16}} .$$
	Besides if for $x\in\R$, we denote $x_+ = \max(0,x)$, then
	\begin{align*}
	&\P\left( \norm{P_n^* U^* A\otimes I_n U P_n} \geq \normt{A} + \varepsilon \right) \\
	&\leq \P\Big( \norm{P_n^* U^* A\otimes I_n U P_n} - \E\left[\norm{P_n^* U^* A\otimes I_n U P_n}\right] \\
	&\quad\quad\quad \geq \varepsilon - \E\Big[ \Big( \norm{P_n^* U^* A\otimes I_n U P_n} - \normt{A}\Big)_+ \Big] \Big) \\
	&\leq \P\Big( \left| \norm{P_n^*U^* A\otimes I_n UP_n} - \E\Big[ \norm{P_n^*U^* A\otimes I_n UP_n} \Big] \right| \\
	&\quad\quad\quad \geq \varepsilon - \E\Big[ \Big( \norm{P_n^* U^* A\otimes I_n U P_n} - \normt{A}\Big)_+ \Big] \Big) \\
	&\leq 2 e^{-kn\left( \varepsilon - \E\left[ \Big( \norm{P_n^* U^* A\otimes I_n U P_n} - \normt{A}\Big)_+ \right] \right)^2 / 16} .
	\end{align*}
	Besides thanks to our first estimate, i.e. Lemma \ref{4firstestimate}, we get that for any $r>0$,
	\begin{align*}
	\E\left[ \Big( \norm{P_n^* U^* \cdot A\otimes I_n\cdot U P_n} - \normt{A}\Big)_+ \right] &\leq r + \int_r^1 \P\left( \norm{P_n^* U^* A\otimes I_n U P_n} \geq \normt{A} + \alpha \right) d\alpha\\
	&\leq r + 2^{23}\times \frac{\ln^2(kn)}{kn} \int_r^1 3\times \alpha^{-4} d\alpha \\
	&\leq r + 2^{23}\times \frac{\ln^2(kn)}{kn} r^{-3}
	\end{align*}
	And after fixing $ r = \left(2^{23}\times \frac{\ln^2(kn)}{kn}\right)^{1/4}$, we get that 
	\begin{align*}
	\E\left[ \Big( \norm{P_n^* U^* \cdot A\otimes I_n\cdot U P_n} - \normt{A}\Big)_+ \right] &\leq \left(2^{27}\times \frac{\ln^2(kn)}{kn}\right)^{1/4} .
	\end{align*}
	Hence if $ kn \geq 2^{31}\times \ln^2(kn) \times \varepsilon^{-4}$, we have 
	\begin{align*}
	&\P\left( \norm{P_n^* U^* \cdot A\otimes I_n\cdot U P_n} \geq \normt{A} + \varepsilon \right) \leq 2 e^{-kn\times \frac{\varepsilon^2}{64}} .
	\end{align*}
	
\end{proof}

We can finally prove Theorem \ref{4apr} by using the former lemma in combination with Lemma \ref{4ptilemme}.

\begin{proof}[Proof of Theorem \ref{4apr}]
	If we set $u = \frac{\sqrt{2}\varepsilon}{3k^2}$, then with 
	$\mathcal{S}_u = \{ u M\  |\ M\in\M_k(\C)_{sa},\ \forall i\geq j,\ \Re(m_{i,j})\in \left\{\N+ \frac{1}{2}\right\}\cap [0,\lceil u^{-1} \rceil],\ \forall i>j,\ \Im(m_{i,j})\in \left\{\N+ \frac{1}{2}\right\}\cap [0,\lceil u^{-1} \rceil]  \}$, 
	Lemma \ref{4ptilemme} tells us that
	
	$$ \P\left(K_{n,k,t} \not\subset K_{k,t} + \varepsilon\right) \leq \sum_{M\in \mathcal{S}_u} \P\left( \norm{P_n^* U^* (P_{\D_k}M\otimes I_n) U P_n} > \normt{P_{\D_k}M} + \frac{\varepsilon}{3k} \right) .$$
	But thanks to Lemma \ref{4sndestimate}, we know that for any $A\in\D_k$, if 
	$\ n \geq 3^4 \times 2^{31}\times \ln^2(kn) \times k^3 \varepsilon^{-4}$, then
	\begin{equation*}
	\P\left( \norm{P_n^* U^* A\otimes I_n U P_n} \geq \normt{A} + \frac{\varepsilon}{3k} \right) \leq 2 e^{-\frac{n}{k}\times \frac{\varepsilon^2}{576}} .
	\end{equation*}
	
	Thus since the cardinal of $\mathcal{S}_u$ can be bounded by $(u^{-1}+1)^{k^2}$, we get that for 
	$\ n \geq 3^4 \times 2^{31}\times \ln^2(kn) \times k^3 \varepsilon^{-4}$,
	$$ \P\left(K_{n,k,t} \not\subset K_{k,t} + \varepsilon\right) \leq 2 (u^{-1}+1)^{k^2} e^{-\frac{n}{k}\times \frac{\varepsilon^2}{576}} \leq e^{k^2(\ln(3k^2 \varepsilon^{-1})) -\frac{n}{k}\times \frac{\varepsilon^2}{576} } . $$
	
\end{proof}

\section{Application to Quantum Information Theory}\label{4sec:QIT}

\subsection{Preliminaries on entropy}

For $X\in D(H)$, its \emph{von Neumann entropy} is defined by functional
calculus by $H(X)=-\tr(X\ln X)$, where $0\ln 0$ is assumed by continuity to be zero.
In other words, $H(X)=\sum_{\lambda\in spec(X)}-\lambda\ln\lambda$ where the
sum is counted with multiplicity. A \emph{quantum channel} $\Phi : B(H_1)\to B(H_2)$ is a completely positive trace preserving linear map. 
The \emph{Minimum Output Entropy} (MOE) of $\Phi$ is
\begin{equation}\label{4def:MOE}
H_{min}(\Phi)=\min_{X\in D(H_1)} H(\Phi (X)).
\end{equation}
During the last decade, a crucial problem in Quantum Information Theory was to determine whether one can find two quantum channels
$$\Phi_i: B(H_{j_i})\to B(H_{k_i}), i=\{1,2\},$$ such that 
$$H_{min}(\Phi_1\otimes\Phi_2)<H_{min}(\Phi_1)+H_{min}(\Phi_2).$$

If $x\in\R^k$, we define $\normt{x}$ as the $t$-norm of the diagonal matrix of size $k$ whose coefficients are those of $x$. Then let $e_1 = (1, 0, \ldots, 0) \in \R^k$ and let
\begin{equation}\label{4eq:def-xopt}
x_t^*=\left(\normt{e_1}, \underbrace{\frac{1-\normt{e_1}}{k-1}, \ldots, \frac{1-\normt{e_1}}{k-1}}_{k-1 \text{ times}}\right).
\end{equation}
If we view $x_t^*$ as a diagonal matrix, then it can be easily checked that $ x_t^*\in K_{k,t}$, and the following is the main result of \cite{BCN2}:
\begin{theorem}
	For any $p> 1$, the maximum of the $l^p$ norm on $K_{k,t}$ is reached at the point
	$x_t^*$.
\end{theorem}
By letting $p\to 1$ it implies that the minimum of the entropy on $K_{k,t}$ is reached at the point $x_t^*$ and this is what we will be using. 
For the sake of making actual computation, it will be useful to recall the value of $\normt{e_1}$. For this, we 
use the following notation:
\begin{equation}
(1^j0^{k-j}) = (\underbrace{1, 1, \ldots, 1}_{j \text{ times}},\underbrace{0, 0, \ldots, 0}_{k-j \text{ times}}) \in \R^k,
\end{equation}
and $1^k = (1^k0^0)$.
It was proved in the early days of 
free probability theory (see \cite{voiculescu-dykema-nica}) that for $j=1, 2, \ldots, k$, one has
\begin{equation}
\label{4formulexp}
\normt{(1^j0^{k-j})}=\phi(u,t) :=
\begin{cases}
t +u-2tu +2\sqrt{tu (1-t)(1-u)}  &\text{ if } t + u < 1,\\
1  &\text{ if } t + u \geq 1,
\end{cases}
\end{equation}
where $u = j/k$.

\subsection{Corollary of the main result}

The following is a direct consequence of the main theorem in terms of possible entropies of the output set.

\begin{theorem}
	\label{4presque}
	With $S_{n,k,t} = \min_{A\in K_{n,k,t}} H(A) = H_{\min}(\qc)$ and $S_{k,t} = \min_{A\in K_{k,t}} H(A)$, if we assume $d_n\leq tkn$, then for $\ n \geq 3^4 \times 2^{31}\times \ln^2(kn) \times k^3 \varepsilon^{-4}$ where $0< \varepsilon \leq e^{-1}$,
	$$ \P\Big(S_{n,k,t} \leq S_{k,t} - 3k \varepsilon |\ln(\varepsilon)| \Big) \leq e^{k^2(\ln(3k^2 \varepsilon^{-1})) -\frac{n}{k}\times \frac{\varepsilon^2}{576} } . $$
\end{theorem}

\begin{proof}
	Let $A,B \in \D_k$ such that $\norm{A-B}_2 \leq \varepsilon$ with $\norm{M }_2  = \sqrt{\tr_{k}(M^*M)}$, with eigenvalues $(\lambda_i)_i$ and $(\mu_i)_i$. Then with $\widetilde{x} = \max \{\varepsilon,x\}$,
	\begin{align*}
	\Big| |\tr_{k}(A\ln(A))| - |\tr_{k}(B\ln(B))| \Big| &= \Big| \sum_i \lambda_i \ln(\lambda_i) - \sum_i \mu_i \ln(\mu_i) \Big| \\
	&\leq 2k \sup_{x\in [0,\varepsilon]} |x\ln(x)| + \Big| \sum_i \widetilde{\lambda_i} \ln(\widetilde{\lambda_i}) - \sum_i \widetilde{\mu_i} \ln(\widetilde{\mu_i}) \Big| \\
	&\leq 2k \varepsilon |\ln(\varepsilon)| + \sum_i |\lambda_i -\mu_i| |\ln(\varepsilon)| \\
	&\leq 2k \varepsilon |\ln(\varepsilon)| +k \norm{A-B} |\ln(\varepsilon)| \\
	&\leq 3k \varepsilon |\ln(\varepsilon)|
	\end{align*}
	Thus with the notation $ K_{k,t} + \varepsilon$ as in Theorem \ref{4apr}, if $K_{n,k,t} \subset K_{k,t} + \varepsilon$, then 
	$$S_{n,k,t} \geq \min_{A\in K_{k,t} + \varepsilon} |\tr_{k}(A\ln(A))| \geq S_{k,t} - 3k \varepsilon |\ln(\varepsilon)|.$$
	 Hence 
	$$ \P\Big(S_{n,k,t} \leq S_{k,t} - 3k \varepsilon |\ln(\varepsilon)| \Big) \leq \P\left(K_{n,k,t} \not\subset K_{k,t} + \varepsilon\right). $$
	Theorem \ref{4apr} then allows us to conclude.
	
\end{proof}

\subsection{Application to violation of the Minimum Output Entropy of Quantum Channels}

In order to obtain violations for the additivity relation of the minimum output entropy, one needs to obtain upper bounds for the quantity $H_{\min}(\Phi \otimes \Psi)$ for some quantum channels $\Phi$ and $\Psi$. The idea of using conjugate 
channels ($\Psi  = \bar \Phi$) and bounding the minimum output entropy by the value of the entropy 
at the Bell state dates back to Werner, Winter and others (we refer to \cite{hayden-winter} for references). 
To date, it has been proven to be the most successful method of 
tackling the additivity problem. The following inequality is elementary 
and lies at the heart of the method:
\begin{equation}
H_{\min} (\Phi \otimes \bar \Phi) \leq H([ \Phi \otimes \bar \Phi] (E_{d})),
\end{equation}
where $E_{d}$ is the maximally entangled state over the input space $(\C^d)^{\otimes 2}$. 
More precisely, $E_d$ is the projection on the Bell vector
\begin{equation}
Bell_d = \frac{1}{\sqrt{d}}\sum_{i=1}^{d} e_i \otimes e_i,
\end{equation}
where $\{e_i\}_{i=1}^{d}$ is a fixed basis of $\C^{d}$.

For random quantum channels $\Phi = \Phi_n$, the random output matrix $[\Phi_n\otimes \bar \Phi_n] (E_{d_n})$
was thoroughly studied in \cite{CN1} in the regime $d_n \sim tkn$; we recall here one of the main results of that paper.
There, it was proved that 
almost surely, as $n$ tends to infinity, the random matrix $[\Phi_n\otimes\bar \Phi_n](E_{d_n}) \in M_{k^2}(\C)$ 
has eigenvalues
\begin{equation}\label{4eq:probabilistic-bound}
\gamma^*_t = \left( t + \frac{1-t}{k^2},\underbrace{\frac{1-t}{k^2}, \ldots, \frac{1-t}{k^2}}_{k^2-1 \text{ times}}\right).
\end{equation}
This result improves on a bound  \cite{hayden-winter} via linear algebra techniques,
which states that the largest eigenvalue of the random matrix $[\Phi_n\otimes \bar \Phi_n](E_{d_n})$
is at least $d_n/(kn)\sim t$.
Although it might be possible to work directly with the bound provided by \eqref{4eq:probabilistic-bound} 
with additional probabilistic consideration,
for the sake of concreteness we will work with the bound of  \cite{hayden-winter}.
% For this purpose, let us consider two non-increasing sequences $\lambda= (\lambda_1,\ldots ), \mu=(\mu_1,\ldots)$ 
% of non-negative numbers whose sum is $1$. We say that $\lambda$ \emph{majorizes} $\mu$
% iff for any integer $k$, $\lambda_1+\ldots +\lambda_k\ge \mu_1+\ldots +\mu_k$.
% It is known that in this case, $H(\lambda )\le H(\mu )$. %preciser la signification (claire) de l'entropie dans ce cas-la. 
Thus if the largest eigenvalue of $[\Phi_n\otimes \bar \Phi_n](E_{d_n})$ is $d_n/(kn)$, since $\tr_k\otimes\tr_{k}([\Phi_n\otimes \bar \Phi_n](E_{d_n})) = 1$, the entropy is maximized if we take the remaining $k^2-1$ eigenvalues equal to $\frac{1-d_n/(kn)}{k^2-1}$, thus it follows that $$H_{\min} (\Phi \otimes \bar \Phi) \leq H([ \Phi \otimes \bar \Phi] (E_{d_n})) \le H
\left( \frac{d_n}{kn} ,\underbrace{\frac{1-\frac{d_n}{kn}}{k^2-1}, \ldots, \frac{1-\frac{d_n}{kn}}{k^2-1}}_{k^2-1 \text{ times}}\right)
$$

\noindent Therefore, as we will see more clearly in the proof of theorem \ref{4thm:numerics}, it is enough to find $n,k,d_n,t$ such that
\begin{equation}
\label{4delta}
-\frac{d_n}{kn}\log\left(\frac{d_n}{kn}\right) -\left(1-\frac{d_n}{kn}\right)\log \left[\left(1-\frac{d_n}{kn}\right)/(k^2-1)\right]<2H (x_t^*).
\end{equation}
In \cite{BCN2} it was proved with the assistance of a computer that this can be done for any $k\ge 184$, as long as we take $t$ 
around $1/10$, see figure 1 from \cite{BCN2}. However for $k$ large enough, the difference between the right and left term 
of \eqref{4delta} is maximal for $t=1/2$. For example, we obtain the following theorem 

\begin{theorem}\label{4thm:numerics}
	For the following values $(k,t,n) = (184,1/10, 10^{53}), (185,1/10, 10^{52}), (200,1/10, 10^{48}),$ $(500 ,1/10 , 10^{47}), (500 ,1/2 ,  10^{46})$ violation of additivity (i.e. the existence of two quantum channels which satisfy the inequality \eqref{4ineqqMOE}) is achieved with probability at least $1 - \exp(-10^{20})$. 
\end{theorem}

\begin{proof}
	We make sure to work with $n$ a multiple of $10$ so that we can set $d_n = t kn$, then since $H_{\min}(\qc) = H_{\min}(\bar\Phi_n)$,
	\begin{align*}
	&\P\Big( H_{\min}(\Phi_n\otimes\bar\Phi_n) < H_{\min}(\qc) + H_{\min}(\bar\Phi_n) \Big) \\
	&= \P\Big( H_{\min}(\Phi_n\otimes\bar\Phi_n) < 2H_{\min}(\qc) \Big) \\
	&\geq \P\Big( -t \log(t) - (1-t)\log\left[ (1-t)/(k^2-1) \right] < 2H_{\min}(\qc) \Big) \\
	&= 1 - \P\left( H_{\min}(\qc) \leq - \frac{t}{2} \log(t) - \frac{1-t}{2}\log\left( \frac{1-t}{k^2-1} \right) \right) \\
	&= 1 - \P\left( S_{n,k,t} \leq S_{k,t} - \delta_{k,t} \right),
	\end{align*}
	
	\noindent with
	
	$$ \delta_{k,t} = \frac{t}{2} \log(t) + \frac{1-t}{2}\log\left( \frac{1-t}{k^2-1} \right) - \normt{e_1} \log(\normt{e_1}) - (1-\normt{e_1}) \log\left( \frac{1-\normt{e_1}}{k-1} \right) .$$
	
	\noindent Then we conclude with Theorem \ref{4presque} and equation \eqref{4formulexp} to compute explicit parameters.
	
\end{proof}

Let us remark that in \cite{BCN2} what was actually proved is that violation of additivity can occur for any $k\ge 183$.
However, for the output dimension $k=183$, a probabilistic argument is needed -- namely, the limiting distribution of the output of a Bell state
and in turn, the fact that one can give a good estimate in probability for $\liminf_n H_{\min} (\Phi_n \otimes \bar \Phi_n)$
(again, we refer to  \cite{CN1,BCN2} for details). However this estimate is difficult to evaluate explicitly as a function of $n$, therefore, in this paper,
we replace it by a slightly weaker estimate that is always true and that yields violation for any $k\ge 184$. In other words we lose one dimension. 
To conclude, since our bound is explicit, we solve the problem of supplying actual input dimensions for any valid output dimension, for which
the violation of MOE will occur. From a point of view of theoretical probability, this is a step towards a large deviation principle. 
And although our bound is not optimal, our results presumably give the right speed of deviation. However conjecturing a complete large deviation principle and a rate function seems to be beyond the scope of our techniques.  

\bibliographystyle{abbrv}

\end{document}